\documentclass[11pt,a4paper]{article}
\usepackage[]{amsmath,amssymb,amsfonts,latexsym,amsthm,enumerate,a4wide,fullpage,url}

\newcommand*\rfrac[2]{{}^{#1\!\!}/_{\!#2}}
\newcommand*\dist{\text{dist}}
\newcommand*{\bigcupdot}{\ensuremath{\mathop{\makebox[-2pt]{\hspace{14pt}{\(\cdot\)}}\bigcup}}}

\numberwithin{equation}{section}
\newtheorem{thm}{Theorem}[section]

\newtheorem{lem}[thm]{Lemma}
\newtheorem{cor}[thm]{Corollary}
\newtheorem{dfn}[thm]{Definition}

\begin{document}

\title{Walking on the Edge and Cosystolic Expansion}
\author{Tali Kaufman\thanks{Bar-Ilan University, ISRAEL. Email: \texttt{kaufmant@mit.edu}.
Research supported in part by ERC and BSF.}
\and David Mass\thanks{Bar-Ilan University, ISRAEL. Email: \texttt{dudimass@gmail.com}.}}
\maketitle

\begin{abstract}
Random walks on regular bounded degree expander graphs have numerous applications.
A key property of these walks is that they converge rapidly to the uniform distribution on the vertices.
The recent study of expansion of high dimensional simplicial complexes, which are the high dimensional analogues of graphs, calls for the natural generalization of random walks to higher dimensions.
In particular, a high order random walk on a $2$-dimensional simplicial complex moves at random between neighboring edges of the complex, where two edges are considered neighbors if they share a common triangle.
We show that if a regular $2$-dimensional simplicial complex is a cosystolic expander and the underlying graph of the complex has a spectral gap larger than $1/2$, then the random walk on the edges of the complex converges rapidly to the uniform distribution on the edges.
\end{abstract}

\section{Introduction}
Random walks on bounded degree expander graphs have numerous applications in many fields (see~\cite{HLW06} for an excellent survey).
In many theoretical and practical computational problems it is necessary to draw samples uniformly at random from some huge space.
It is well known that a random walk on an expander graph converges rapidly to its stationary distribution, which is the uniform distribution in the case of regular graphs.
When the graph is of a bounded degree it is possible to simulate this random walk efficiently.
Thus, one can efficiently obtain a set of samples, which looks like it has been chosen uniformly and independently at random, by simulating a random walk on a bounded degree expander graph.

In recent years the study of expansion in higher dimensions has emerged, introducing high dimensional simplicial complexes as the analogues of graphs in higher dimensions.
A $2$-dimensional simplicial complex can be viewed as a hypergraph that contains vertices, edges and triangles, with a closure property.
Namely, if a triangle is in the hypergraph, then all of its edges and vertices are also in the hypergraph.

The well studied random walk can be generalized naturally to higher dimensions by walking on the edges of the complex instead of walking on the vertices, where two edges are considered neighbors if they share a common triangle.
If the complex is regular in the dimension of its edges, i.e., every edge is contained in the same number of triangles, then the stationary distribution of this random walk is the uniform distribution on the edges.
A natural question to be asked, which we address in this work, is the following:

\paragraph{Question.}{\em How fast does a random walk on the edges of a regular $2$-dimensional simplicial complex converge to the uniform distribution on the edges}.
\bigskip

In the case of graphs, it is well known that the speed of convergence of the random walk is controlled by the spectral expansion of the graph.
Though, when moving to higher dimensions, the high dimensional version of spectral expansion is not known to imply the fast convergence of the described random walk.

For a $2$-dimensional simplicial complex $X$, we define the {\em edge-graph} of the complex, denoted by $G_1(X)$, as follows.
The vertices of $G_1(X)$ are the edges of $X$, and there is an edge between two vertices in $G_1(X)$ if the corresponding edges in $X$ share a common triangle.
It is easy to see that a random walk on the edges of $X$ is exactly the same as a random walk on the vertices of $G_1(X)$.
Thus, if the edge-graph is a spectral expander, then the high order random walk on the complex converges rapidly to the uniform distribution on the edges.

The question we address here is what are the requirements on the complex which imply the expansion of its edge-graph.

In a recent work~\cite{KM16} it has been shown that if the complex is composed of excellent spectral expander graphs, i.e., every small piece of the entire complex is a spectral expander by itself, then the high order random walk converges rapidly to the uniform distribution.
In this work we show that in the $2$-dimensional case it is enough for the global complex to be a combinatorial expander (in its high dimensional analogue), for the high order random walk to converge rapidly to the uniform distribution on the edges.
This global expansion property might hold even without the expansion of every piece of the complex.
The only additional requirement for the rapid convergence of this random walk is that the underlying graph of the complex (which is obtained by ignoring the triangles of the complex) has a spectral gap larger than $\rfrac{1}{2}$.
This is a reasonable assumption as it is already known that there are {\em bounded degree} high dimensional complexes which have this property (see~\cite{KKL14, EK15} for an explicit construction).

One of the motivations behind this high order random walk is the following.
In the case of expander graphs, one can obtain a uniformly random vertex by exploring {\em all} possible walks of some small length from any initial vertex of the graph.
Now, consider the underlying graph of a $2$-dimensional simplicial complex for which the random walk on its edges converges rapidly to the uniform distribution.
In this case, in order to obtain a uniformly random vertex, it is enough to explore only walks for which their subsequent edges share a common triangle.
The length of the walks to be explored might be slightly larger than in the previous case, but many of these walks can be ignored since they do not induce a walk on the edges of the complex.
Namely, we achieve a way to sample a uniformly random vertex of the graph by considering only a {\em subset} of all the possible walks of some small length.

We achieve our result by showing that the edge-graph of the complex is a combinatorial expander, and then we use the relation between combinatorial and spectral expansion of graphs in order to deduce that the edge-graph is a spectral expander (which controls the speed of convergence of the walk on the edges of the complex).
For the edge-graph to be a combinatorial expander we need to show that every subset of vertices has many outgoing edges.
The method we use, in essence, is the following.
Consider a $2$-dimensional simplicial complex $X=(V,E,T)$, and its edge-graph $G_1(X)=(V_{G_1},E_{G_1})$ (where $V_{G_1}=E$).
For any subset of edges $F \subseteq E$ and any vertex $v \in V$ in the complex, we define the {\em local view} of $v$ with regard to $F$ to be the edges of $F$ for which one of their endpoints is $v$.
For every subset of vertices in the edge-graph $S \subseteq V_{G_1}$ we consider the corresponding subset of edges in the complex $F \subseteq E$.
We show a relation between the outgoing edges of $S$ and the expansion of the corresponding edges in the complex.
Specifically, the number of the outgoing edges of $S$ can be calculated by considering the expansion of all the local views of vertices in $X$ with regard to the corresponding edges $F$.
Each local view is a subset of edges of the complex, and thus by the global expansion property of the complex it has to expand a lot.
The only subtlety is that very large local views are not guaranteed to expand so well.
We show that if the underlying graph of the complex has a spectral gap larger than $\rfrac{1}{2}$, then there must be many local views which are not too large.
Thus, in this case, every subset of vertices in the edge-graph has many outgoing edges, which implies that the edge-graph is a combinatorial expander as required.

\subsection{Expander graphs}
We recall some basic properties of expander graphs.
Throughout this section $G=(V,E)$ is a $k$-regular undirected graph on $n$ vertices, i.e., every vertex is contained in exactly $k$ edges and $|V| = n$.
The {\em Cheeger constant} of $G$, denoted by $h(G)$, is defined as
$$h(G) = \min_{\substack{S \subseteq V \\ 0 < |S| \le \frac{n}{2}}}\frac{|E(S,\bar S)|}{|S|},$$
where $E(S,\bar S)$ is the set of edges with one endpoint in $S$ and one endpoint in $\bar S$.
Note that $0 \le h(G) \le k$, thus we denote its normalized value by $\widetilde h(G) = h(G)/k$.

The graph $G$ is said to be an {\em $\epsilon$-combinatorial expander} if $\widetilde h(G) \ge \epsilon$, for some $\epsilon > 0$.
A {\em family of graphs} $\{G_i\}_{i \in \mathbb N}$ is called a family of $\epsilon$-combinatorial expanders if there exists a constant $\epsilon > 0$, such that every graph in the family is an $\epsilon$-combinatorial expander.

The adjacency matrix of $G$, denoted by $A = A(G)$, is the symmetric $n \!\times\! n$ matrix, where $A(u,v) = 1$ if $\{u,v\} \in E$, otherwise $A(u,v) = 0$.
Denote by $\lambda_1(G) \ge \lambda_2(G) \ge \dotsb \ge \lambda_n(G)$ the eigenvalues of $A$.
For ease of notation we will just write $\lambda_i$ instead of $\lambda_i(G)$ where the graph $G$ is clear from the context.
It is well known that for every $k$-regular graph $\lambda_1 = k$.
It is also known that $\lambda_n = -k$ if and only if $G$ has a bipartite component, i.e., a disconnected component which can be partitioned into two subsets of vertices such that there are no edges inside of each subset.
Denote by $\widetilde \lambda_1 \ge \widetilde \lambda_2 \ge \dotsb \ge \widetilde \lambda_n$ the normalized eigenvalues of $A$, where $\widetilde \lambda_i = \lambda_i / k$ for every $1 \le i \le n$.
The {\em spectral gap} of $G$ is defined as $1 - \widetilde \lambda_2$, which captures the gap between the trivial eigenvalue of $A$ and the largest non-trivial eigenvalue.
Denote by $\widetilde \lambda = \max\{|\widetilde \lambda_2|, |\widetilde \lambda_n|\}$ the largest non-trivial eigenvalue of $A$ in absolute value.

The graph $G$ is said to be an {\em $\epsilon$-spectral expander} if $\widetilde \lambda \le \epsilon$, for some $0 < \epsilon < 1$.
A family of graphs is called a family of $\epsilon$-spectral expanders if there exists a constant $0 < \epsilon < 1$, such that every graph in the family is an $\epsilon$-spectral expander.

\subsection{Random walks}
For any vertex $v \in V$, denote by $N(v) = \{u \in V \;|\; \{u,v\} \in E\}$ the neighbors of $v$.
A random walk on a graph is a sequence of vertices $v_0, v_1, \dotsc \in V$, such that,
\begin{enumerate}
  \item $v_0$ is chosen from some initial probability distribution on the vertices.
  \item For every $i \ge 1$, the vertex $v_i$ is chosen uniformly at random from $N(v_{i-1})$.
\end{enumerate}

Denote by $\mathbf p_0 \in \mathbb R^n$ the initial probability distribution, and by $\mathbf p_i \in \mathbb R^n$ the probability distribution after $i$ steps of the walk, i.e., $\mathbf p_i(v)$ is the probability to find the walker on vertex $v$ after $i$ steps of the walk.
Denote by $\mathbf u = (\rfrac{1}{n}, \dotsc, \rfrac{1}{n})$ the uniform distribution.
We define a rapid mixing of the random walk as follows.
\begin{dfn}[Rapid mixing]
Let $G=(V,E)$ be a $k$-regular graph.
The random walk on $G$ is said to be {\em $\alpha$-rapidly mixing}, for some $\alpha > 0$, if for any initial probability distribution $\mathbf p_0 \in \mathbb R^n$ and any $i \in \mathbb N$,
$$\|\mathbf p_i - \mathbf u\|_2 \le \alpha^i.$$
\end{dfn}

The following Theorem is a very useful and well known property of expander graphs.
\begin{thm}\label{rapid-mixing-thm}
Let $G=(V,E)$ be a $k$-regular graph.
If $G$ is an $\epsilon$-spectral expander, then the random walk on $G$ is $\alpha$-rapidly mixing for $\alpha=\epsilon$.
\end{thm}

\subsection{$2$-dimensional simplicial complexes}
We present here some basic definitions of $2$-dimensional simplicial complexes.
A $2$-dimensional simplicial complex $X=(V,E,T)$ contains vertices, edges and triangles, denoted by $X(0) = V$, $X(1) = E$ and $X(2) = T$, respectively, with a closure property.
Namely, for every triangle $\{u,v,w\} \in T$, all of its edges exist in the complex, i.e., $\{u,v\},\{u,w\},\{v,w\} \in E$, and for every edge $\{u,v\} \in E$, both of its endpoints are in the complex, i.e., $u,v \in V$.
The vertices of the complex are called {\em faces} of dimension $0$, the edges are the $1$-dimensional faces, and the triangles are the $2$-dimensional faces.
The complex is termed {\em bounded degree} if the number of edges and triangles incident to each vertex is bounded by some constant, independent of the number of vertices in the complex.
The complex is said to be $(k_0, k_1)$-regular if every vertex is contained in exactly $k_0$ edges and every edge is contained in exactly $k_1$ triangles.

For any subset of vertices $S \subseteq V$, its {\em coboundary}, denoted by $\delta(S)$, is the set of edges with one endpoint in $S$ and one endpoint in $\bar S$, i.e., $\delta(S) = E(S, \bar S)$.
The {\em $0$-coboundaries}, denoted by $B^0(X)$, are the subsets of vertices for which their coboundary is always $0$, i.e., $B^0(X) = \{\emptyset, V\}$.
These are also called the {\em trivial zeros}.
The {\em $0$-cocycles}, denoted by $Z^0(X)$, are all the subsets of vertices for which their coboundary is $0$, i.e., $Z^0(X) = \{S \subseteq V \;|\; |\delta(S)| = 0\}$.
Note that $B^0(X) \subseteq Z^0(X)$, with equality if and only if all the vertices of the complex are connected.

For any subset of edges $F \subseteq E$, its {\em coboundary}, denoted by $\delta(F)$, is the set of triangles with an {\em odd} number of edges in $F$.
To make this definition clearer, we view $F$ as a function $F : E \rightarrow \{0,1\}$, where $F(\{u,v\}) = 1$ if and only if $\{u,v\} \in F$.
Then,
$$\delta(F) = \{\{u,v,w\} \in T \;|\; F(\{u,v\}) + F(\{u,w\}) + F(\{v,w\}) = 1 \mkern-10mu\mod 2\}.$$
The {\em $1$-coboundaries}, denoted by $B^1(X)$, are the subsets of edges for which their coboundary is always $0$, or again, the trivial zeros.
These can be identified with the {\em cuts} in the complex, i.e., all the edges between any partition of the vertices into two subsets, or analogously, $B^1(X) = \{\delta(S) \;|\; S \subseteq V\}$.
The {\em $1$-cocycles}, denoted by $Z^1(X)$, are all the subsets of edges for which their coboundary is $0$, i.e., $Z^1(X) = \{F \subseteq E \;|\; |\delta(F)| = 0\}$.
Note again that $B^1(X) \subseteq Z^1(X)$.

\subsection{$2$-dimensional expanders}
High order combinatorial expansion has received much attention recently.
There are two mainly studied variants of combinatorial expansion in higher dimensions.
They both generalize the Cheeger constant of graphs to high dimensional simplicial complexes.
We present both of them just for $2$-dimensional regular complexes, which are enough for our purposes.

Recall that in graphs, combinatorial expansion is defined with relation to the Cheeger constant of the graph.
An equivalent definition is the following.
Let $G=(V,E)$ be a $k$-regular graph.
For any two subsets of vertices $S,T \subseteq V$, the {\em distance} of $S$ from $T$ is defined as $\dist(S, T) = |S \setminus T| + |T \setminus S|$.
The distance of $S \subseteq V$ from a set of subsets of vertices $T_1, \dotsc, T_m \subseteq V$ is defined as $\dist(S, \{T_1, \dotsc, T_m\}) = \min_{1 \le i \le m}\dist(S, T_i)$.
The graph $G$ is said to be an $\epsilon$-combinatorial expander, for $\epsilon > 0$, if for every subset $S \subseteq V$, the following holds:
\begin{itemize}
  \item If $|E(S, \bar S)| = 0$, then $S$ is a {\em trivial non-expanding set}, i.e., $S \in \{\emptyset, V\}$.
  \item Otherwise, $\frac{|E(S, \bar S)|}{\dist(S, \{\emptyset, V\})} = \frac{|E(S, \bar S)|}{\min\{|S|, |\bar S|\}} \ge \epsilon k$.
\end{itemize}

In other words, the only non-expanding sets of $G$ are the trivial non-expanding sets, which are either the empty set or the set of all the vertices, and every other set has to expand with proportion to its size.
Note that this definition is equivalent to the common definition of combinatorial expansion of graphs, which was presented above.

The first generalization of combinatorial expansion to higher dimensions is termed {\em coboundary expansion}~\cite{LM06}.
Let $X=(V,E,T)$ be a $2$-dimensional $(k_0, k_1)$-regular simplicial complex.
$X$ is said to be an {\em $\epsilon$-coboundary expander}, for $\epsilon > 0$, if for every $i \in \{0,1\}$, and every subset of $i$-dimensional faces $S \subseteq X(i)$, the following holds:
\begin{itemize}
  \item If $|\delta(S)| = 0$, then $S$ is a trivial zero, i.e., $S \in B^i(X)$.
  \item Otherwise, $\frac{|\delta(S)|}{\dist(S, B^i(X))} \ge \epsilon k_i$.
\end{itemize}

This definition is a natural generalization of the combinatorial expansion of graphs, under the observation that for any subset of vertices $S \subseteq V$, $\delta(S) = E(S, \bar S)$.
The existence of bounded degree expanders according to this definition is not known! (either random or explicit).

The next studied combinatorial expansion in higher dimensions is termed {\em cosystolic expansion}~\cite{G10, EK15}.
This definition is a relaxation of the previous one;
it allows non-expanding sets to be non-trivial, as long as they are very large.
Let $X=(V,E,T)$ be a $2$-dimensional $(k_0, k_1)$-regular simplicial complex.
$X$ is said to be an {\em $(\epsilon, \mu)$-cosystolic expander}, for $\epsilon, \mu > 0$, if for every $i \in \{0,1\}$, and every subset of $i$-dimensional faces $S \subseteq X(i)$, the following holds:
\begin{itemize}
  \item If $|\delta(S)| = 0$, then either $S$ is a trivial zero, i.e., $S \in B^i(X)$, or $S$ is very large, i.e., $|S| \ge \mu|X(i)|$.
  \item Otherwise, $\frac{|\delta(S)|}{\dist(S, Z^i(X))} \ge \epsilon k_i$.
\end{itemize}

Note that in this case, since there are non-expanding sets which are not trivial, then the distance is taken from the cocycles of the complex, which contain non-trivial non-expanding sets.
This definition of expansion is more natural to the world of computer science as it can be viewed as a property testing question, where the property is whether a given set is a cocycle (see~\cite{KL14} for more on high dimensional expansion and property testing).
Moreover, this definition of expansion implies a property called the topological overlapping of a complex, which was heavily studied~\cite{G10, DKW15}.
The first known explicit bounded degree high dimensional simplicial complexes according to this definition have been constructed in~\cite{KKL14, EK15}.

\subsection{Our contribution}
In this work we study further the high order random walk defined in~\cite{KM16}.
We show that the walk on the edges of a regular $2$-dimensional cosystolic (or coboundary) expander converges rapidly to the uniform distribution on the edges.
We achieve our result by a new method which relates the high order combinatorial expansion of the complex to the combinatorial expansion of its edge-graph.
Then we use the relation between spectral and combinatorial expansion of graphs in order to deduce the concentration of the spectrum of the edge-graph's adjacency matrix, which implies the rapid convergence of the high order random walk.

\subsubsection{High order random walks}
Let $X=(V,E,T)$ be a $2$-dimensional $(k_0, k_1)$-regular simplicial complex.
For any edge $e \in E$, denote by $N(e) = \{f \in E \;|\; e \cup f \in T\}$ the neighbors of $e$.
The high order random walk on $X$ is a sequence of edges $e_0, e_1, \dotsc \in E$, such that,
\begin{enumerate}
  \item $e_0$ is chosen from some initial probability distribution on the edges.
  \item For every $i \ge 1$, the edge $e_i$ is chosen uniformly at random from $N(e_{i-1})$.
\end{enumerate}

The high order random walk is said to be {\em $\alpha$-rapidly mixing} if for any initial probability distribution on the edges $\mathbf p_0 \in \mathbb R^{|E|}$ and any $i \in \mathbb N$,
$$\|\mathbf p_i - \mathbf u\|_2 \le \alpha^i.$$
where $\mathbf p_i \in \mathbb R^{|E|}$ is the probability distribution after $i$ steps of the walk.

\subsubsection{Edge-graph and underlying graph}
For a $2$-dimensional simplicial complex $X=(V,E,T)$, denote by $G_0(X) = (V,E)$ the underlying graph of the complex, which is obtained by simply ignoring the triangles of the complex.
We define the {\em edge-graph} of $X$, denoted by $G_1(X)$, to be the graph that its vertices are the edges of $X$, and there is an edge between two vertices if the corresponding edges in $X$ are neighbors.
Note that $G_0(X)$ and $X$ have the same set of vertices, while the vertices of $G_1(X)$ are different; they are the edges of $X$.
It is clear that if $G_1(X)$ is a spectral expander, then the high order random walk on $X$ converges rapidly to the uniform distribution on the edges (recall Theorem \ref{rapid-mixing-thm}).
We prove the following Theorem.
\begin{thm}[Main Theorem, informal, for formal see Theorem \ref{main-thm}]
If $X$ is a $2$-dimensional cosystolic expander and $G_0(X)$ has a spectral gap larger than $\rfrac{1}{2}$, then $G_1(X)$ is a spectral expander, and hence the high order random walk on $X$ mixes rapidly.
\end{thm}

\subsection{Related work}
In a recent work of Parzanchevski and Rosenthal~\cite{PR12} a notion of high order {\em topological} random walk on high dimensional simplicial complexes has been studied.
The topological random walk was designed to expose the topological properties of the complex.
Parzanchevski and Rosenthal define a variant of the stationary distribution of the random walk and show its relation to the spectrum of the high order laplacian on the space that is {\em orthogonal to the coboundaries}.
In our work, the stationary distribution of the random walk is already known to be the uniform distribution.
Moreover, it is already known that the speed of convergence is controlled by the concentration of the spectrum of the high order adjacency matrix on the space that is {\em orthogonal to the constant functions}.
We show here that cosystolic expansion of a complex implies the concentration of the spectrum of its high order adjacency matrix, a thing that can not be deduced in any way from~\cite{PR12}.

In a previous work of the authors of this work~\cite{KM16} the speed of convergence of the high order random walk on simplicial complexes of every dimension has been studied.
It has been shown that if all the links of a simplicial complex (i.e., all the pieces which compose the entire complex) are excellent spectral expanders, then the high order random walk on it converges rapidly to its stationary distribution.
Here we show that for the $2$-dimensional case, the rapid convergence of the high order random walk can be deduced from the global expansion property of the complex, which might hold even without the assumption that every small piece of it expands.
Moreover, while in~\cite{KM16} the concentration of the spectrum of the high order adjacency matrix was deduced through a new notion of high dimensional combinatorial expansion, which is termed colorful expansion, here we show how to deduce it {\em directly} from the cosystolic expansion of the complex.
Thus, the method used here is different from the one used in~\cite{KM16}.

\section{Preliminaries}
We present here some preliminaries that are needed for the next section.

A well known pseudorandom property of expander graphs is captured by the Expander Mixing Lemma, proven by Alon and Chung~\cite{AC06}.
We state here one side of it, which states that for every subset of vertices in an expander graph, the number of edges for which both of their endpoints is in the set is approximately the expected number as in a random graph.
\begin{lem}[Expander Mixing Lemma]\label{expander-mixing-lemma}
Let $G=(V,E)$ be a $k$-regular graph.
For any subset of vertices $S \subseteq V$,
$$2|E(S)| \le k|S|\!\left(\frac{|S|}{|V|} + \widetilde \lambda_2(G)\!\right),$$
where $E(S) = E(S,S)$ is the set of edges with both endpoints in $S$.
\end{lem}

The relation between combinatorial and spectral expansion of graphs is known as the Cheeger's inequality.
We state here one side of it, proven by Alon~\cite{A86}, which lets us bound the second largest eigenvalue of a graph by its Cheeger constant.
\begin{lem}[Cheeger's inequality]\label{cheeger-inequality}
For any $k$-regular graph $G$,
$$\widetilde \lambda_2(G) \le 1 - \frac{\widetilde h(G)^2}{2}.$$
\end{lem}

In order to bound the smallest eigenvalue of a graph we use a Cheeger-type inequality, proven by Trevisan~\cite{T16}, which defines a combinatorial measure of how far is a graph from having a bipartite component and relates the smallest eigenvalue of the graph to this measure.
We state here a specific case of the result proven in~\cite{KM16}, which uses this inequality, in order to bound the smallest eigenvalue of the edge-graph.
\begin{lem}[Cheeger-type inequality for edge-graph]\cite[Lemma~5.2]{KM16}\label{cheeger-type-inequality}
For any $2$-dimensional $(k_0, k_1)$-regular simplicial complex $X$,
$$\widetilde \lambda_n(G_1(X)) \ge -\frac{17}{18}.$$
(it is actually enough for the complex to be regular only in the dimension of its edges, i.e., every edge is contained in the same number of triangles.
We write $(k_0, k_1)$-regular in order to avoid confusing notations).
\end{lem}

\section{Main Theorem}
In this section we prove our main theorem, which is the following.

\begin{thm}\label{main-thm}
Let $X=(V,E,T)$ be a $2$-dimensional $(k_0, k_1)$-regular $(\epsilon, \mu)$-cosystolic expander.
If $\widetilde \lambda_2 = \widetilde \lambda_2(G_0(X)) < \rfrac{1}{2}$, then the high order random walk on the edges of $X$ is $\alpha$-rapidly mixing for
$$\alpha = 1 - \frac{\epsilon^2}{128}\left(3\sqrt{(1+2\widetilde\lambda_2)^2 + 32} - 2\widetilde\lambda_2 - 17 \right)^{\!\!2}.$$
\end{thm}

\subsection{Proof sketch}
Since the high order random walk on the edges of $X$ is equivalent to a random walk on the edge-graph $G_1(X)$, it is sufficient to show that the edge-graph is a spectral expander.
For that we need to bound both $\widetilde \lambda_2(G_1(X))$ and $\widetilde \lambda_n(G_1(X))$.
The bound on $\widetilde \lambda_n(G_1(X))$ is obtained by the structure of the edge-graph, which is far from having a bipartite component.
The bound on $\widetilde \lambda_2(G_1(X))$ is derived from the combinatorial expansion of the edge-graph, which is the main focus of our proof.

In order to prove that the edge-graph is a combinatorial expander we need to show that every non-empty subset of vertices of size at most half of the vertices has many outgoing edges.
At Lemma \ref{outgoing-edges-lemma} we prove an equivalence between the outgoing edges of any subset of vertices in the edge-graph and the coboundaries of all the local views with regard to the corresponding edges in the complex.
Thus, it is enough to show that for every subset of edges in the complex, there are many local views for which their coboundary is large.

For any subset of edges in the complex, we split the vertices of the complex to three subsets according to their local view's size:
Vertices with a very large local view are called fat, vertices with a large local view but not too large are called semi-fat, and vertices with a small local view are called non-fat.
At Lemma \ref{local-view-coboundary-lemma} we provide a lower bound on the coboundary of the local view of every semi-fat and non-fat vertex.
This lemma follows from the cosystolic expansion of $X$ and the spectral expansion of the underlying graph of the complex $G_0(X)$;
the cosystolic expansion of the complex guarantees a large coboundary relative to the distance from the non-expanding sets, and the spectral expansion of the underlying graph provides a large distance from the non-expanding sets.
Then we split to two cases:
If there are many semi-fat vertices, then we immediately get enough local views for which their coboundary is large.
Otherwise, almost all of the vertices are fat.
By the spectral expansion of $G_0(X)$ we know that there are not many edges for which both of their endpoints are fat vertices.
Thus, there must be many edges touching non-fat vertices, which implies that many non-fat vertices have a large local view.
Then again we conclude that many local views have a large coboundary.

\subsection{Outgoing edges and coboundaries of local views}
First let us define the local view of a vertex in a formal way.
\begin{dfn}[Local view]
Let $X=(V,E,T)$ be a $2$-dimensional simplicial complex.
For any subset of edges $F \subseteq E$ and any vertex $v \in V$, the {\em local view} of $v$ with regard to $F$ is defined as
$$F_v = \{e \in F \;|\; v \in e\}.$$
\end{dfn}

We start with the following lemma, which provides a way to count the outgoing edges of any subset of vertices in the edge-graph by summing over the coboundaries of all the local views with regard to the corresponding edges in the complex.

\begin{lem}\label{outgoing-edges-lemma}
Let $X=(V,E,T)$ be a $2$-dimensional simplicial complex, and let $G_1 = G_1(X) = (V_{G_1}, E_{G_1})$ be its edge-graph.
For any subset of vertices in the edge-graph, $S \subseteq V_{G_1}$, denote by $F \subseteq E$ the corresponding subset of edges in $X$. Then
$$|E_{G_1}(S, \bar S)| = \sum_{v \in V}|\delta(F_v)|.$$
\end{lem}

\begin{proof}
First recall that every vertex $v_{G_1} \in V_{G_1}$ corresponds to an edge in $X$, and thus it is of the form $v_{G_1} = \{u,v\} \in E$, where $u,v \in V$.
Next note that every edge $\{u_{G_1},v_{G_1}\} \in E_{G_1}$ satisfies $u_{G_1} \cap v_{G_1} = v \in V$.
For any vertex $v \in V$, define $E_{G_1}^v = \{ \{u_{G_1},v_{G_1}\} \in E_{G_1} \;|\; u_{G_1} \cap v_{G_1} = v \}$.
Then note that $E_{G_1}$ can be decomposed into the following disjoint union.
\begin{equation}\label{outgoing-edges-eq1}
E_{G_1} = \bigcupdot_{v \in V}E_{G_1}^v.
\end{equation}

Let $S \subseteq V_{G_1}$ be a subset of vertices in $G_1$, and denote by $F \subseteq E$ the corresponding edges in $X$.
By (\ref{outgoing-edges-eq1}) the following holds.
\begin{equation}\label{outgoing-edges-eq2}
|E_{G_1}(S, \bar S)| = \sum_{v \in V}|E_{G_1}^v(S, \bar S)|,
\end{equation}
where $E_{G_1}^v(S, \bar S) = E_{G_1}(S, \bar S) \cap E_{G_1}^v$.
Thus, it is enough to show that $|E_{G_1}^v(S, \bar S)| = |\delta(F_v)|$ for every vertex $v \in V$.

Let $v \in V$.
Note that every edge $\{u_{G_1},v_{G_1}\} \in E_{G_1}^v(S, \bar S)$ satisfies $S(u_{G_1}) + S(v_{G_1}) = 1$.
Denote $u_{G_1} = \{u,v\} \in E$ and $v_{G_1} = \{v,w\} \in E$, so the induced triangle by $u_{G_1}$ and $v_{G_1}$ is $\{u,v,w\} \in T$.
We claim that $\{u,v,w\} \in \delta(F_v)$.
The reason is that $F_v(\{u,v\}) + F_v(\{v,w\}) = S(u_{G_1}) + S(v_{G_1}) = 1$, and $F_v(\{u,w\}) = 0$ since $v \notin \{u,w\}$.
Thus,
\begin{equation}\label{outgoing-edges-eq3}
|E_{G_1}^v(S, \bar S)| \le |\delta(F_v)|.
\end{equation}

Similarly, consider a triangle $\{u,v,w\} \in \delta(F_v)$.
Since $F_v(\{u,w\}) = 0$, it must hold that $F_v(\{u,v\}) + F_v(\{v,w\}) = 1$.
Denote $\{u,v\} = u_{G_1} \in V_{G_1}$ and $\{v,w\} = v_{G_1} \in V_{G_1}$, so the induced edge in $G_1$ is $\{u_{G_1},v_{G_1}\} \in E_{G_1}^v$.
We claim that $\{u_{G_1},v_{G_1}\} \in E_{G_1}^v(S, \bar S)$.
The reason is that $S(u_{G_1}) + S(v_{G_1}) = F_v(\{u,v\}) + F_v(\{v,w\}) = 1$.
Thus,
\begin{equation}\label{outgoing-edges-eq4}
|\delta(F_v)| \le |E_{G_1}^v(S, \bar S)|.
\end{equation}

By (\ref{outgoing-edges-eq3}) and (\ref{outgoing-edges-eq4}) we conclude that $|E_{G_1}^v(S, \bar S)| = |\delta(F_v)|$ for every vertex $v \in V$.
Then applying it to (\ref{outgoing-edges-eq2}) finishes the proof.
\end{proof}

\subsection{Local views of semi-fat and non-fat vertices}
The purpose of this subsection is to prove that the local views of semi-fat and non-fat vertices expand well, i.e., that the coboundary of their local view is large.
First, let us define these vertices formally.

\begin{dfn}[Fat, semi-fat, and non-fat vertices]
Let $X=(V,E,T)$ be a $2$-dimensional $(k_0, k_1)$-regular simplicial complex,
let $\rfrac{1}{2} < \eta < 1$ be a fatness constant,
and let $F \subseteq E$ be a subset of edges.
For any vertex $v \in V$,
\begin{itemize}
  \item If $|F_v| > \eta k_0$ then $v$ is called a {\em fat} vertex.
  \item If $\frac{k_0}{2} < |F_v| \le \eta k_0$ then $v$ is called a {\em semi-fat} vertex.
  \item If $|F_v| \le \frac{k_0}{2}$ then $v$ is called a {\em non-fat} vertex.
\end{itemize}
\end{dfn}

Now we need the following lemmas, which let us evaluate the distance of the local view of any vertex in the complex from the cocycles of the complex.
We prove the following lemmas asymptotically, i.e., for $|V|$ large enough.

\begin{lem}[Large cuts]\label{large-cuts-lemma}
Let $G=(V,E)$ be a $k$-regular graph.
If $\widetilde \lambda_2 = \widetilde \lambda_2(G) < \rfrac{1}{2}$, then for any subset of vertices $\emptyset \ne S \subsetneq V$, $|E(S, \bar S)| \ge k$.
\end{lem}

\begin{proof}
Let $\emptyset \ne S \subsetneq V$.
If $|S| = 1$ then $|E(S, \bar S)| = k$ and the proof is done.
So assume that $|S| \ge 2$.
We can also assume that $|S| \le \rfrac{|V|}{2}$, otherwise we replace $S$ with $\bar S$.
Since $G$ is $k$-regular, the following holds.
\begin{equation}\label{large-cuts-eq-1}
k|S| = 2|E(S)| + |E(S, \bar S)|.
\end{equation}

By the Expander Mixing Lemma (Lemma~\ref{expander-mixing-lemma}),
\begin{equation}\label{large-cuts-eq-2}
2|E(S)| \le k|S|\biggr(\frac{|S|}{|V|} + \widetilde \lambda_2\biggr).
\end{equation}

Combining both (\ref{large-cuts-eq-1}) and (\ref{large-cuts-eq-2}) yields the following lower bound on $|E(S, \bar S)|$.
\begin{equation}\nonumber
|E(S, \bar S)| \ge k|S|\biggr(1 - \frac{|S|}{|V|} - \widetilde \lambda_2\biggr).
\end{equation}

Define $f(|S|) = |S|(1 - \frac{|S|}{|V|} - \widetilde \lambda_2)$.
The claim holds if $f(|S|) \ge 1$ for every $2 \le |S| \le \rfrac{|V|}{2}$.
We show that it holds for $|V|$ large enough.
In particular, assume that $|V| \ge \frac{4}{1-2\widetilde\lambda_2}$.
Since $f$ is a concave function it is enough to check its extreme values.
Indeed, for $|S| = 2$,
\begin{equation}\nonumber
f(2) = 2\biggr(1 - \frac{2}{|V|} - \widetilde \lambda_2\biggr) \ge 2\biggr(1 - \frac{1-2\widetilde\lambda_2}{2} - \widetilde\lambda_2\biggr) = 1,
\end{equation}
and for the other extreme value,
\begin{equation}\nonumber
f\biggr(\frac{|V|}{2}\biggr) = \frac{|V|}{2}\biggr(1 - \frac{1}{2} - \widetilde \lambda_2\biggr) = |V|\biggr(\frac{1-2\widetilde\lambda_2}{4}\biggr) \ge 1.
\end{equation}
\end{proof}

\begin{lem}[Large non-trivial non-expanding sets]\label{distance-from-cosystoles-lemma}
Let $X=(V,E,T)$ be a $2$-dimensional $(k_0, k_1)$-regular $(\epsilon, \mu)$-cosystolic expander.
For any subset of edges $F \subseteq E$ and any vertex $v \in V$, $\dist\,(F_v, Z^1(X) \!\setminus\! B^1(X)) \ge \frac{k_0}{2}$.
\end{lem}

\begin{proof}
Let $F \subseteq E$ and let $v \in V$.
For any $z \in Z^1(X) \!\setminus\! B^1(X)$,
\begin{equation}\nonumber
\dist(F_v, z) = |F_v \setminus z| + |z \setminus F_v| \ge |z| - |F_v| \ge |z| - k_0.
\end{equation}

Thus, the claim holds if $|z| \ge \frac{3}{2}k_0$ for every $z \in Z^1(X) \!\setminus\! B^1(X)$.
We show that it holds for $|V|$ large enough.
In particular, assume that $|V| \ge \frac{3}{\mu}$.
Then for every $z \in Z^1(X) \!\setminus\! B^1(X)$,
\begin{equation}\nonumber
|z| \ge \mu|E| = \mu\frac{k_0|V|}{2} \ge \mu\frac{3k_0}{2\mu} = \frac{3}{2}k_0,
\end{equation}
where the first inequality holds since $z$ is a non-trivial non-expanding set of edges, and thus by the cosystolic expansion of $X$ its size must be at least $\mu |E|$.
\end{proof}

\begin{cor}[Distance of local view from the non-expanding sets]\label{distance-from-cocycles-lemma}
Let $X=(V,E,T)$ be a $2$-dimensional $(k_0, k_1)$-regular $(\epsilon, \mu)$-cosystolic expander.
If $\widetilde \lambda_2(G_0(X)) < \rfrac{1}{2}$, then for any subset of edges $\emptyset \ne F \subsetneq E$ and any vertex $v \in V$, $\mbox{dist}\,(F_v, Z^1(X)) = \min\{|F_v|,\, k_0-|F_v|\}$.
\end{cor}

\begin{proof}
Let $\emptyset \ne F \subsetneq E$, and let $v \in V$.
Note that for $S \in \{\emptyset, V\}$,
\begin{equation}\label{distance-from-cocycles-eq1}
\dist(F_v, \delta(S)) = \dist(F_v, \emptyset) = |F_v|.
\end{equation}

For any other subset of vertices $\emptyset \ne S \subsetneq V$, by Lemma \ref{large-cuts-lemma}, $|\delta(S)| = |E(S,\bar S)| \ge k_0$, and thus,
\begin{equation}\label{distance-from-cocycles-eq2}
\dist(F_v, \delta(S)) = |F_v \setminus \delta(S)| + |\delta(S) \setminus F_v| \ge |\delta(S)| - |F_v| \ge k_0 - |F_v|,
\end{equation}
where equality holds for $S = \{v\}$.

Recall that the trivial non-expanding sets are $B^1(X) = \{\delta(S) \;|\; S \subseteq V\}$.
Thus, by (\ref{distance-from-cocycles-eq1}) and (\ref{distance-from-cocycles-eq2}) we can calculate the distance of $F_v$ from the trivial non-expanding sets.
\begin{equation}\nonumber
\dist(F_v, B^1(X)) = \min_{S \subseteq V}\dist(F_v, \delta(S)) = \min\{|F_v|,\, k_0 - |F_v|\}.
\end{equation}

It is left to show that the distance of $F_v$ from the non-trivial non-expanding sets is larger than its distance from the trivial non-expanding sets.
Indeed, by Lemma \ref{distance-from-cosystoles-lemma},
\begin{equation}\nonumber
\dist(F_v, Z^1(X) \!\setminus\! B^1(X)) \ge \frac{k_0}{2} \ge \min\{|F_v|,\, k_0 - |F_v|\} = \dist(F_v, B^1(X)).
\end{equation}

Thus, the distance of $F_v$ from all the non-expanding sets is its distance from trivial non-expanding sets, i.e.,
\begin{equation}\nonumber
\dist(F_v, Z^1(X)) = \dist(F_v, B^1(X)) = \min\{|F_v|,\, k_0 - |F_v|\}.
\end{equation}
\end{proof}

Now we can prove the following lemma, which shows that the local views of semi-fat and non-fat vertices have a large coboundary.

\begin{lem}[Coboundaries of local views of semi-fat and non-fat vertices]\label{local-view-coboundary-lemma}
Let $X=(V,E,T)$ be a $2$-dimensional $(k_0, k_1)$-regular $(\epsilon, \mu)$-cosystolic expander, let $\rfrac{1}{2} < \eta < 1$ be a fatness constant, and let $F \subseteq E$ be a subset of edges.
If $\widetilde \lambda_2(G_0(X)) < \rfrac{1}{2}$ then,
\begin{itemize}
  \item For any semi-fat vertex $v \in V$, $|\delta(F_v)| \ge \epsilon k_1(1-\eta)k_0$.
  \item For any non-fat vertex $v \in V$, $|\delta(F_v)| \ge \epsilon k_1|F_v|$.
\end{itemize}
\end{lem}

\begin{proof}
Let $v \in V$.
If $v$ is a semi-fat vertex, then by Corollary~\ref{distance-from-cocycles-lemma},
\begin{equation}\nonumber
\dist(F_v, Z^1(X)) = \min\{|F_v|,\, k_0 - |F_v|\} = k_0 - |F_v| \ge (1 - \eta)k_0,
\end{equation}
where the last equality holds since $|F_v| > \frac{k_0}{2}$, and the last inequality holds since $|F_v| \le \eta k_0$.

Thus, by the cosystolic expansion of $X$ we get,
\begin{equation}\nonumber
|\delta(F_v)| \ge \epsilon \,k_1\,\dist(F_v, Z^1(X)) \ge \epsilon k_1(1 - \eta)k_0.
\end{equation}

If $v$ is a non-fat vertex, then by Corollary~\ref{distance-from-cocycles-lemma},
\begin{equation}\nonumber
\dist(F_v, Z^1(X)) = \min\{|F_v|,\, k_0 - |F_v|\} = |F_v|,
\end{equation}
where the last equality holds since $|F_v| \le \frac{k_0}{2}$.

Thus, by the cosystolic expansion of $X$ we get,
\begin{equation}\nonumber
|\delta(F_v)| \ge \epsilon \,k_1\,\dist(F_v, Z^1(X)) = \epsilon k_1|F_v|.
\end{equation}
\end{proof}

\subsection{Lower bound on the sum of coboundaries}
The following is the main technical lemma of our proof, which provides a lower bound on the sum of the coboundaries of the local views of all the vertices in the complex.

\begin{lem}[Sum of coboundaries]\label{sum-of-coboundaries-lemma}
Let $X=(V,E,T)$ be a $2$-dimensional $(k_0, k_1)$-regular $(\epsilon, \mu)$-cosystolic expander.
If $\widetilde \lambda_2 = \widetilde \lambda_2(G_0(X)) < \rfrac{1}{2}$, then for any subset of edges $F \subseteq E$ such that $\frac{|F|}{|E|} \le \rfrac{1}{2}$,
$$\sum_{v \in V}|\delta(F_v)| \ge \frac{\epsilon k_1}{4}\left(3\sqrt{(1+2\widetilde\lambda_2)^2 + 32} - 2\widetilde\lambda_2 - 17 \right)|F|.$$
\end{lem}

\begin{proof}
Let $\rfrac{1}{2} < \eta < 1$ be a fatness constant which will be defined later.
Let $F \subseteq E$ be a subset of edges such that $\frac{|F|}{|E|} \le \rfrac{1}{2}$.
Define $A,B,C \subseteq V$ as the sets of fat, semi-fat and non-fat vertices, respectively.

First we bound the size of $A$.
On the one hand, the local view of each vertex in $A$ is very large, and thus,
\begin{equation}\label{bound-fat-1}
\sum_{v \in A}|F_v| > \sum_{v \in A}\eta k_0 = \eta k_0|A|.
\end{equation}

On the other hand,
\begin{equation}\label{bound-fat-2}
\sum_{v \in A}|F_v| \le \sum_{v \in V}|F_v| = 2|F| = 2|F|\frac{|E|}{|E|} = k_0|V|\frac{|F|}{|E|},
\end{equation}
where the last equality holds since $|E| = \frac{k_0|V|}{2}$.

From (\ref{bound-fat-1}) and (\ref{bound-fat-2}) we get the following upper bound on the size of $A$.
\begin{equation}\nonumber
\frac{|A|}{|V|} < \eta^{-1}\frac{|F|}{|E|}.
\end{equation}

Now we split to two cases according to the number of the semi-fat vertices.

(i) Assume that $\frac{|B|}{|V|} > (\eta^{-1}-1)\frac{|F|}{|E|}$.
By Lemma \ref{local-view-coboundary-lemma} we obtain the following lower bound on the sum of the coboundaries of the local views of all the vertices in the complex.
\begin{equation}\label{bound-sum-semi-fat}
\begin{split}
\sum_{v \in V}|\delta(F_v)| &\ge
\sum_{v \in B}|\delta(F_v)| \ge
\sum_{v \in B}\epsilon k_1(1 - \eta)k_0 >
\epsilon k_1(1 - \eta)k_0(\eta^{-1}-1)\frac{|V||F|}{|E|} \\[5pt] &=
2\epsilon k_1(1 - \eta)(\eta^{-1}-1)|F| =
\epsilon k_1(2\eta^{-1} + 2\eta - 4)|F|.
\end{split}
\end{equation}

(ii) In this case, $\frac{|B|}{|V|} \le (\eta^{-1}-1)\frac{|F|}{|E|}$.
Thus, $\frac{|A \cup B|}{|V|} < (2\eta^{-1} - 1)\frac{|F|}{|E|}$.
Then by the Expander Mixing Lemma (Lemma~\ref{expander-mixing-lemma}) we can bound the number of edges with both endpoints in the sets of fat or semi-fat vertices.
\begin{equation}\nonumber
\begin{split}
|E(A \cup B)| &\le
\frac{k_0|A \cup B|}{2} \left(\frac{|A \cup B|}{|V|} + \widetilde\lambda_2 \right) \\[5pt] &\le
\frac{k_0(2\eta^{-1}-1)|F||V|}{2|E|} \left(\!(2\eta^{-1}-1)\frac{|F|}{|E|} + \widetilde\lambda_2 \right) \\[5pt] &\le
(2\eta^{-1}-1) \left(\frac{2\eta^{-1}-1}{2} + \widetilde\lambda_2 \right)|F|,
\end{split}
\end{equation}
where the last inequality holds since $\frac{k_0|V|}{2} = |E|$ and $\frac{|F|}{|E|} \le \rfrac{1}{2}$.

Since the number of edges with both endpoints in $A \cup B$ is bounded, then there are at least $\Big(\!1 - (2\eta^{-1}-1) \big(\frac{2\eta^{-1}-1}{2} + \widetilde\lambda_2 \big)\!\Big)|F|$ edges for which at least one of their endpoints is in $C$.
Thus,
\begin{equation}\nonumber
\sum_{v \in C}|F_v| \ge \biggr(\!1 - (2\eta^{-1}-1) \biggr(\frac{2\eta^{-1}-1}{2} + \widetilde\lambda_2 \biggr)\!\!\biggr)|F|.
\end{equation}

Then, by Lemma \ref{local-view-coboundary-lemma} we obtain the following lower bound on the sum of the coboundaries of the local views of all the vertices in the complex.
\begin{equation}\label{bound-sum-non-fat}
\begin{split}
\sum_{v \in V}|\delta(F_v)| &\ge
\sum_{v \in C}|\delta(F_v)| \ge
\sum_{v \in C}\epsilon k_1|F_v| =
\epsilon k_1\sum_{v \in C}|F_v| \\[5pt] &\ge
\epsilon k_1 \biggr(\!1 - (2\eta^{-1}-1) \biggr(\frac{2\eta^{-1}-1}{2} + \widetilde\lambda_2 \biggr)\!\!\biggr)|F|.
\end{split}
\end{equation}

Now set the following fatness constant.
$$\eta = \frac{1}{8}\left(1+2\widetilde\lambda_2 + \sqrt{(1+2\widetilde\lambda_2)^2 + 32}\right).$$

Note that since $0 < \widetilde\lambda_2 < \rfrac{1}{2}$ then $\frac{1+\sqrt{33}}{8} < \eta < 1$.
Also note that $\widetilde\lambda_2 = 2\eta - \eta^{-1} - \rfrac{1}{2}$.
Then (\ref{bound-sum-non-fat}) can be simplified as follows.
\begin{equation}\nonumber
\begin{split}
\sum_{v \in V}|\delta(F_v)| &\ge
\epsilon k_1 \biggr(\!1 - (2\eta^{-1}-1) \biggr(\frac{2\eta^{-1}-1}{2} + \widetilde\lambda_2 \biggr)\!\!\biggr)|F| \\ &=
\epsilon k_1 \biggr(\!1 - (2\eta^{-1}-1) \biggr(\frac{2\eta^{-1}-1}{2} + 2\eta - \eta^{-1} - \rfrac{1}{2} \biggr)\!\!\biggr)|F| \\[8pt] &=
\epsilon k_1 \left(1 - (2\eta^{-1}-1) (2\eta - 1)\right)|F| =
\epsilon k_1(2\eta^{-1} + 2\eta - 4)|F|,
\end{split}
\end{equation}
which is exactly the same as (\ref{bound-sum-semi-fat}).

Thus, in both cases the following holds.
\begin{equation}\nonumber
\begin{split}
\sum_{v \in V}|\delta(F_v)| &\ge
\epsilon k_1(2\eta^{-1} + 2\eta - 4)|F| =
\epsilon k_1(\widetilde\lambda_2 + 3\eta^{-1} - \rfrac{7}{2})|F| \\[-4pt] &=
\epsilon k_1\left(\widetilde\lambda_2 + \frac{24}{1+2\widetilde\lambda_2 + \sqrt{(1+2\widetilde\lambda_2)^2 + 32}} - \frac{7}{2}\right)\!|F| \\[5pt] &=
\epsilon k_1\left(\widetilde\lambda_2 - \frac{24(1+2\widetilde\lambda_2 - \sqrt{(1+2\widetilde\lambda_2)^2 + 32})}{32} - \frac{7}{2}\right)\!|F| \\[5pt] &=
\frac{\epsilon k_1}{4}\left(3\sqrt{(1+2\widetilde\lambda_2)^2 + 32} - 2\widetilde\lambda_2 - 17 \right)\!|F|,
\end{split}
\end{equation}
which finishes the proof.
\end{proof}

\subsection{Proof of Theorem \ref{main-thm}}
In order to prove the theorem we need to show that $G_1 = G_1(X) = (V_{G_1}, E_{G_1})$ is a spectral expander graph.
We show first that $G_1$ is a combinatorial expander graph, which then implies a bound on $\widetilde \lambda_2(G_1)$.

Let $S \subseteq V_{G_1}$ be a subset of vertices in the edge-graph such that $0 < \frac{|S|}{|V_{G_1}|} \le \rfrac{1}{2}$, and denote by $F \subseteq E$ the corresponding subset of edges in $X$.
By Lemma \ref{outgoing-edges-lemma} we know that $|E_{G_1}(S, \bar S)| = \sum_{v \in V}|\delta(F_v)|$.
Since $\frac{|F|}{|E|} = \frac{|S|}{|V_{G_1}|} \le \rfrac{1}{2}$, we can apply Lemma \ref{sum-of-coboundaries-lemma} and get following.
\begin{equation}\nonumber
\begin{split}
|E_{G_1}(S, \bar S)| &=
\sum_{v \in V}|\delta(F_v)| \ge
\frac{\epsilon k_1}{4}\left(3\sqrt{(1+2\widetilde\lambda_2)^2 + 32} - 2\widetilde\lambda_2 - 17 \right)\!|F| \\[5pt] &=
\frac{\epsilon k_1}{4}\left(3\sqrt{(1+2\widetilde\lambda_2)^2 + 32} - 2\widetilde\lambda_2 - 17 \right)\!|S|.
\end{split}
\end{equation}

Note that since each edge of $X$ is contained in $k_1$ triangles, then $G_1$ is $2k_1$-regular.
Thus, we get the following lower bound on the normalized Cheeger constant of $G_1$.
\begin{equation}\nonumber
\widetilde h(G_1) =
\min_{\substack{S \subseteq V_{G_1} \\ 0 < \frac{|S|}{|V_{G_1}|} \le \rfrac{1}{2}}}\frac{|E_{G_1}(S, \bar S)|}{2k_1|S|} \ge
\frac{\epsilon}{8}\left(3\sqrt{(1+2\widetilde\lambda_2)^2 + 32} - 2\widetilde\lambda_2 - 17 \right).
\end{equation}

Then by the Cheeger's inequality (Lemma~\ref{cheeger-inequality}) we get the following upper bound on $\widetilde \lambda_2(G_1)$.
\begin{equation}\nonumber
\widetilde \lambda_2(G_1) \le
1 - \frac{\widetilde h(G_1)^2}{2} \le
1 - \frac{\epsilon^2}{128}\left(3\sqrt{(1+2\widetilde\lambda_2)^2 + 32} - 2\widetilde\lambda_2 - 17 \right)^{\!\!2}.
\end{equation}

The bound on $\widetilde \lambda_n(G_1)$ is derived from the structure of the graph $G_1$, which is far from having a bipartite component.
By Lemma~\ref{cheeger-type-inequality},
\begin{equation}\nonumber
\widetilde \lambda_n(G_1) \ge -\frac{17}{18}.
\end{equation}

Since $|\widetilde \lambda_2(G_1)| \ge |\widetilde \lambda_n(G_1)|$, then $G_1$ is an $\alpha$-spectral expander for $\alpha = \widetilde \lambda_2(G_1)$.
Thus, by Theorem~\ref{rapid-mixing-thm} the random walk on $G_1$ is $\alpha$-rapidly mixing, which implies that the high order random walk on the edges of $X$ is $\alpha$-rapidly mixing, which finishes the proof.


\begin{thebibliography}{10}

\bibitem{A86}
N.~Alon.
\newblock Eigenvalues and expanders.
\newblock {\em Combinatorica}, 6(2):83--96, 1986.

\bibitem{AC06}
N.~Alon and F.~R.~K. Chung.
\newblock Explicit construction of linear sized tolerant networks.
\newblock {\em Discrete Mathematics}, 306(10-11):1068--1071, 2006.

\bibitem{DKW15}
{D. Dotterrer, T. Kaufman and U. Wagner}.
\newblock {On Expansion and Topological Overlap}.
\newblock arXiv:1506.04558, 2015.

\bibitem{EK15}
S.~Evra and T.~Kaufman.
\newblock {Bounded Degree Cosystolic Expanders of Every Dimension}.
\newblock arXiv:1510.00839, 2015.

\bibitem{G10}
M.~Gromov.
\newblock {Singularities, Expanders and Topology of Maps. Part 2: from
  Combinatorics to Topology Via Algebraic Isoperimetry}.
\newblock {\em Geometric And Functional Analysis}, 20(2):416--526, 2010.

\bibitem{KL14}
T.~Kaufman and A.~Lubotzky.
\newblock High dimensional expanders and property testing.
\newblock In {\em Innovations in Theoretical Computer Science, ITCS'14,
  Princeton, NJ, USA, January 12-14, 2014}, pages 501--506, 2014.

\bibitem{KM16}
T.~Kaufman and D.~Mass.
\newblock {High Dimensional Combinatorial Random Walks and Colorful Expansion}.
\newblock arXiv:1604.02947, 2016.

\bibitem{LM06}
N.~Linial and R.~Meshulam.
\newblock Homological connectivity of random 2-complexes.
\newblock {\em Combinatorica}, 26(4):475--487, 2006.

\bibitem{PR12}
O.~Parzanchevski and R.~Rosenthal.
\newblock Simplicial complexes: spectrum, homology and random walks.
\newblock arXiv:1211.6775, 2012.

\bibitem{HLW06}
{S. Hoory, N. Linial and A. Wigderson}.
\newblock Expander graphs and their applications.
\newblock {\em Bulletin of the American Mathematical Society}, 43(4):439--561,
  2006.

\bibitem{KKL14}
{T. Kaufman, D. Kazhdan and A. Lubotzky}.
\newblock {Isoperimetric Inequalities for Ramanujan Complexes and Topological
  Expanders}.
\newblock arXiv:1409.1397, 2014.

\bibitem{T16}
L.~Trevisan.
\newblock {Cheeger-type Inequalities for $\lambda_n$}.
\newblock
  \url{http://lucatrevisan.wordpress.com/2016/02/09/cheeger-type-inequalities-for-%CE%BBn/},
  2016.

\end{thebibliography}

\end{document}